%% file: main.tex
\documentclass[a4paper,reqno]{amsart}

\usepackage[british]{babel}

\numberwithin{equation}{section}

\usepackage{amssymb}
\usepackage{amsmath,amsfonts,amsthm} 
\usepackage[pdftex]{graphicx}	
\usepackage{url}
\usepackage[title,titletoc,page]{appendix} 
\usepackage{listings} 
\usepackage{color}
\usepackage{comment}
\usepackage{caption}
\usepackage{subcaption}
\usepackage{algorithm}

\usepackage{fancyhdr}
\numberwithin{equation}{section}		
\numberwithin{figure}{section}			
\numberwithin{table}{section}				

\newcommand{\babs}[1]{\left|{#1}\right|}

\newcommand{\vect}[1]{\boldsymbol{\mathbf{#1}}}

\newcommand{\R}{\mathbb{R}}
\newcommand{\N}{\mathbb{N}}

\usepackage[left=3.5cm,right=3.5cm,top=3.5cm,bottom=3.5cm,footskip=1.5cm,headsep=1cm]{geometry}
\usepackage[%
pdfauthor={Habib Ammari, Silvio Barandun, Yannick De Bruijn, Ping Liu, and Clemens Thalhammer},%
pdftitle={PROPERTIES OF $k-$TOEPLITZ MATRICES},
pdfsubject={PROPERTIES OF $k-$TOEPLITZ MATRICES},
pdfkeywords={Non-Hermitian systems, skin effect, condensation, subwavelength resonators, imaginary gauge
potential, phase transition, localization},
]{hyperref}

\input{comands}

\begin{document}

\title[Spectra and pseudo-spectra of tridiagonal $k$-Toeplitz matrices
]{Spectra and pseudo-spectra of tridiagonal $k$-Toeplitz matrices and the topological origin of the non-Hermitian skin effect }

 \author[H. Ammari]{Habib Ammari}
\address{\parbox{\linewidth}{Habib Ammari\\
 ETH Z\"urich, Department of Mathematics, Rämistrasse 101, 8092 Z\"urich, Switzerland.}}
\email{habib.ammari@math.ethz.ch}
\thanks{}

\author[S. Barandun]{Silvio Barandun}
 \address{\parbox{\linewidth}{Silvio Barandun\\
 ETH Z\"urich, Department of Mathematics, Rämistrasse 101, 8092 Z\"urich, Switzerland.}}
 \email{silvio.barandun@sam.math.ethz.ch}

\author[Y. De Bruijn]{Yannick De Bruijn}
 \address{\parbox{\linewidth}{Yannick De Bruijn\\
 ETH Z\"urich, Department of Mathematics, Rämistrasse 101, 8092 Z\"urich, Switzerland.}}
\email{ydebruijn@student.ethz.ch}

\author[P. Liu]{Ping Liu}
 \address{\parbox{\linewidth}{Ping Liu\\
 ETH Z\"urich, Department of Mathematics, Rämistrasse 101, 8092 Z\"urich, Switzerland.}}
\email{ping.liu@sam.math.ethz.ch}

\author[C. Thalhammer]{Clemens Thalhammer}
 \address{\parbox{\linewidth}{Clemens Thalhammer\\
 ETH Z\"urich, Department of Mathematics, Rämistrasse 101, 8092 Z\"urich, Switzerland.}}
\email{cthalhammer@student.ethz.ch}

\maketitle

\begin{abstract}
    We establish new results on the spectra and pseudo-spectra of tridiagonal $k$-Toeplitz operators and matrices. In particular, we prove the connection between the winding number of the eigenvalues of the symbol function and the exponential decay of the associated eigenvectors (or pseudo-eigenvectors). Our results elucidate the topological origin of the non-Hermitian skin effect in general one-dimensional polymer systems of subwavelength resonators with imaginary gauge potentials, proving the observation and conjecture in \cite{dimerSkin}. We also numerically verify our theory for these systems.
\end{abstract}

\vspace{3mm}



\input{Introduction}            
\input{spectrum_tridiag}        
\input{Pseudo_spectra}          
\input{Topological_origin}

\section*{Acknowledgements}
This work was partially supported by Swiss National Science Foundation grant number 200021-200307.

\section*{Code Availability}
The data that support the findings of this work are openly available at \url{https://doi.org/10.5281/zenodo.10438679}.

\section*{Conflict of interest} 
The authors have no competing interests to declare that are relevant to the content of this
article.

\input{Appendix}

\printbibliography

\end{document}



Last, we characterize the condensation of the pseudo-eigenvectors of $\vect A_{N}$.

%% file: comands.tex
\usepackage[T1]{fontenc} 
\usepackage{lmodern} 
\rmfamily 
\DeclareFontShape{T1}{lmr}{b}{sc}{<->ssub*cmr/bx/sc}{}
\DeclareFontShape{T1}{lmr}{bx}{sc}{<->ssub*cmr/bx/sc}{}
\usepackage{lipsum}
\usepackage{mathtools}
\usepackage{amsmath}
\usepackage{amssymb}
\usepackage{tikz}
\usetikzlibrary{matrix,positioning,decorations.pathreplacing,calc}
\usetikzlibrary{
decorations.text,%
decorations.markings,%
shadows}
\usepackage{adjustbox}
\usepackage{graphicx}

\usepackage{dsfont} 
\usepackage[format=hang]{caption}
\usepackage{subcaption}
\usepackage[normalem]{ulem}

\usepackage{bm}

\usepackage[
style=numeric,
sorting=nty,
maxnames=99,
maxalphanames=5,
natbib=true,
sortcites]{biblatex}

\DeclareNameAlias{default}{family-given}

\graphicspath{{figures/}}

\AtEveryBibitem{
 \clearfield{url}
 \clearfield{issn}
 \clearfield{isbn}
 \clearfield{urldate}
 
 \ifentrytype{book}{
  \clearfield{pages}}{%
 }
}

\usepackage{xargs}                      
\usepackage[prependcaption,textsize=tiny,textwidth=3cm]{todonotes}
\newcommandx{\unsure}[2][1=]{\todo[linecolor=red,backgroundcolor=red!25,bordercolor=red,#1]{#2}}
\newcommandx{\change}[2][1=]{\todo[linecolor=blue,backgroundcolor=blue!25,bordercolor=blue,#1]{#2}}
\newcommandx{\info}[2][1=]{\todo[linecolor=OliveGreen,backgroundcolor=OliveGreen!25,bordercolor=OliveGreen,#1]{#2}}
\newcommandx{\improvement}[2][1=]{\todo[linecolor=black,backgroundcolor=black!25,bordercolor=black,#1]{#2}}
\newcommandx{\thiswillnotshow}[2][1=]{\todo[disable,#1]{#2}}
\usepackage{amsthm}
\usepackage[noabbrev,capitalize]{cleveref}
\crefname{proposition}{Proposition}{Propositions}
\crefname{equation}{}{}

\newtheorem{theorem}{Theorem}[section]
\newtheorem{lemma}[theorem]{Lemma}

\theoremstyle{definition}
\newtheorem{definition}[theorem]{Definition}

\newtheorem{remark}[theorem]{Remark}

\crefname{assumption}{Assumption}{Assumptions}
\crefname{definition}{Definition}{Definitions}
\crefname{corollary}{Corollary}{Corollaries}
\crefname{enumi}{item}{items}

\usepackage{fancyhdr}

\fancypagestyle{plain}{%
\fancyhead[C]{}
\fancyfoot[C]{}

}
\fancypagestyle{nosection}{%
\fancyhead[CE]{}
\fancyhead[CO]{}
\fancyfoot[CE,CO]{\thepage}
}

\DeclareMathOperator{\C}{\mathbb{C}}

\renewcommand{\bar}[1]{\overline{#1}}

\renewcommand{\qedsymbol}{$\blacksquare$}

\usepackage{fancyhdr}

\fancypagestyle{plain}{%
\fancyhead[C]{}
\fancyfoot[C]{}

}
\fancypagestyle{nosection}{%
\fancyhead[CE]{}
\fancyhead[CO]{}
\fancyfoot[CE,CO]{\thepage}
}

\renewcommand{\epsilon}{\varepsilon}

\usepackage{tcolorbox}
\usetikzlibrary{tikzmark,calc}

\usepackage{enumerate}
\usepackage{upgreek}

%% file: Introduction.tex
\section{Introduction}

Motivated by the study of the non-Hermitian skin effect in polymer systems of subwavelength resonators and its topological origin, in this paper we extend the theory on the spectra and pseudo-spectra of Toeplitz operators and matrices to tridiagonal $k$-Toeplitz operators and matrices. Our main goal is to establish a connection between, on one hand, the winding number of the eigenvalues of the symbol function of a tridiagonal $k$-Toeplitz operator and, on the other hand, the exponential decay of the associated  eigenvectors or pseudo-eigenvectors. By doing so, we elucidate the topological origin of the eigenmode condensation  at one edge of a polymer system of subwavelength resonators where an imaginary gauge potential is added inside the resonators to break Hermiticity. We show that, due to a nonzero total winding of the eigenvalues of the symbol of the $k$-Toeplitz operator corresponding to the polymer system, the eigenmodes exhibit exponential decay and condensate at one of the edges of the structure.

In addition to advancing our understanding of the non-Hermitian skin effect, 
which can be seen as the classical analogue of the non-Hermitian Anderson model \cite{hatano}, one of our key contributions in this paper is to optimally characterize the spectra of tridiagonal $k$-Toeplitz operators. Notably, we reveal the distinctive property that the eigenvectors of these operators undergo exponential decay. Using the notion of pseudo-spectrum, we rigorously justify the exponential decay of the pseudo-eigenvectors of the associated  
$k$-Toeplitz matrices and significantly generalize the results presented in \cite{gohberg1993basic, LUMSDAINE1998103, REICHEL1992153, trefethen.embree2005Spectra, bottcher.silbermann1999Introduction}.

Our paper is organized as follows. 
In Section \ref{sec: spectral_Theory_operator}, we first provide an optimal characterization of the spectra of tridiagonal $k$-Toeplitz and Laurent operators. Then we prove the existence of an exponentially decaying eigenvector for Topelitz operators. 
In Section \ref{sec: spectral_Theory_matrix}, we establish a connection between  the eigenvectors of a  tridiagonal $k$-Toeplitz operator and the associated tridiagonal $k$-Toeplitz matrix. In Section \ref{Sec: Topological_origin} we discuss the topological origin of the non-Hermitian skin effect in polymer systems of subwavelength resonators and numerically verify our findings on the spectra of tridiagonal $k$-Toeplitz operators.  In Section \ref{sec: concluding_remarks}, we make some concluding remarks and give some possible extensions of our present work.



%% file: spectrum_tridiag.tex
\section{Spectra of tridiagonal $k$-Toeplitz operators}\label{sec: spectral_Theory_operator}

In this section, we characterize the spectrum of tridiagonal $k$-Toeplitz operators and prove that their associated eigenvectors exhibit exponential decay. This type of operators and their associated truncated operators, i.e. matrices, are crucial in the study of the one-dimensional polymer systems of subwavelength resonators \cite{ammari2023mathematical, ammari.davies.ea2021Functional, feppon.ammari2022Subwavelength, dimerSkin}. We define a \emph{tridiagonal $k$-Toeplitz operator} by

\begin{equation}\label{equ:tridiagonalktoeplitzop1}
    A = \begin{pmatrix}
        a_1 & b_1 & 0 & & & &   \\
        c_1 & a_2 & b_2 &  \ddots & \\
        0  & \ddots &  \ddots & \ddots  & \ddots \\
         & \ddots& c_{k-2} & a_{k-1}& b_{k-1} &\ddots \\
        & & \ddots & c_{k-1} & a_k  & b_k& \ddots\\
        & & & \ddots & c_{k} & a_1  & b_1& \ddots\\
        & & & & \ddots & \ddots & \ddots& \ddots \\
    \end{pmatrix},
\end{equation}
where $a_i,b_i,c_i\in\C$ for $i\in\N$ and $A_{ij}=0$ if $\vert i-j \vert > 1$ for $i,j\in\N$. Its finite, truncated, counterpart will be called a \emph{tridiagonal $k$-Toeplitz matrix}. 


To characterize the spectrum of tridiagonal $k$-Toeplitz operators, we introduce the notions of Fredholm operator and Fredholm index. Let $X$ be a Banach space and $\mathcal{B}(X)$ the set of bounded linear operators on $X$. We say that an operator $A \in \mathcal{B}(X)$ is a \emph{Fredholm operator} if $\operatorname{Im} A$ is a closed subspace in $X$ and
\[
\operatorname{dim} \mathrm{Ker}(A) < \infty,\quad  \operatorname{dim} \mathrm{Coker} A < \infty,
\]
where $\mathrm{Coker} A:= X/ \mathrm{Im} A$. The \emph{index} of the Fredholm operator $A$ is defined as 
\begin{equation}\label{defi:fredholmindex1}
\operatorname{Ind} A \coloneqq \operatorname{dimKer}A- \operatorname{dimCoker}A.
\end{equation}

\subsection{Tridiagonal $k$-Toeplitz operators}
Observe that the tridiagonal $k$-Toeplitz operator $A$ presented in (\ref{equ:tridiagonalktoeplitzop1}) can be reformulated as a tridiagonal block Toeplitz operator, where the blocks repeat in a $1$-periodic way:
\begin{equation}\label{equ:tridiagonalktoeplitzop2}
    A = \begin{pmatrix}
        A_0 & A_{-1} &  \\
        A_{1} & A_0 & \ddots  \\
         & \ddots & \ddots   \\
    \end{pmatrix},
\end{equation}
with
\begin{equation*}
    A_0 = \begin{pmatrix}
        a_1 & b_1 & 0 & \cdots &  0  \\
        c_1 & a_2 & b_2 &  \ddots &\vdots \\
        0  & \ddots &  \ddots & \ddots  & 0 \\
        \vdots & \ddots& c_{k-2} & a_{k-1}& b_{k-1}\\
         0&   \cdots & 0 & c_{k-1} & a_k \\       
    \end{pmatrix},  A_{-1} = \begin{pmatrix}
        0 & \cdots&\cdots&  0\\
        \vdots &\ddots & & \vdots \\
        0& &\ddots & \vdots\\
        b_k &0 &\cdots& 0
    \end{pmatrix}, A_{1}=\begin{pmatrix}
        0 & \cdots& 0 & c_k\\
        \vdots &\ddots & & 0 \\
        \vdots& &\ddots & \vdots\\
        0 &\cdots &\cdots& 0
    \end{pmatrix}.
\end{equation*}
The \emph{symbol} of the tridiagonal block Toeplitz operator \eqref{equ:tridiagonalktoeplitzop2} (and thus also of the tridiagonal $k$-Toeplitz operator \eqref{equ:tridiagonalktoeplitzop1})
is defined by 
\begin{align}
    f:S^1&\to \C^{k\times k}\nonumber\\
    z&\mapsto A_{-1}z^{-1} + A_0 + A_1z,
\end{align}
or explicitly
\begin{equation}\label{eq: symbol tridiagonal operator}
    f(z) = 
    \begin{pmatrix}
        a_1 & b_1 & 0 & \dots & 0 & c_k z \\
        c_1 & a_2 & b_2 & 0 & \dots & 0 \\
        0 & c_2 & a_3 & b_3 & \ddots & \vdots \\
        \vdots  & \ddots &  \ddots & \ddots & \ddots & 0 \\
        0 & \dots & 0 & c_{k-2} & a_{k-1}& b_{k-1}\\
        b_kz^{-1} & 0 &  \dots & 0 & c_{k-1} & a_k \\   
    \end{pmatrix}.
\end{equation}
We will write $A=T(f)$. In the following sections, we will characterize the spectrum of the operator $T(f)$ based on its symbol $f(z)$.

\subsection{Hardy Spaces}

To get a better understanding of the connection between the symbol $f$  and the operator $T(f)$ itself, it is necessary to introduce the so-called \emph{Hardy Spaces}. Let $\mathbb T$ be the unit circle in $\mathbb C$.
\begin{definition}
    The Hardy spaces $H^2:=H^2(\mathbb{T})$ and $H^2_-:=H^2_-(\mathbb{T})$ are defined by
    \begin{align*}
        &H^2:=\{f\in L^2(\mathbb{T}): f_n=0 \text{ for } n<0\},\\
        &H^2_-:=\{f\in L^2(\mathbb{T}): f_n=0 \text{ for } n\geq0\},
    \end{align*}
    where $f_n$ denotes the Fourier coefficients of $f$.
\end{definition}
The spaces $H^2$ and $H^2_-$ are closed, orthogonal subspaces of $L^2\coloneqq L^2(\mathbb{T})$. Hence, one may write
\begin{equation*}
    L^2 = H^2\oplus H^2_-.
\end{equation*}
Denote by $P$ the projection of $L^2$ onto $H^2$. The functions 
\begin{equation*}
    \left\{\frac{1}{\sqrt{2\pi}} e^{in\theta}\right \}_{n=0}^\infty
\end{equation*}
form an orthonormal basis of $H^2$. One important property of functions in $H^2$ is given by the following theorem \cite[Theorem 6.13]{Douglas1998}.
\begin{theorem}[F. and M. Riesz]\label{thm: F. and M. Riesz}
    If f is a nonzero function in $H^2$, then the set $\{e^{it} \in \mathbb{T}:
    f(e^{it})=0\}$ has measure zero.
\end{theorem}
It is well-known \cite{trefethen.embree2005Spectra} that for a Toeplitz operator $T(f)$ with symbol $f \in L^\infty$,  $T(f)$ is the matrix representation of the operator
\begin{equation*}
    H^2 \rightarrow H^2, g\mapsto P(fg)
\end{equation*}
in the above orthonormal basis. A similar notion holds for block Toeplitz operators. We denote by $(H^2)^k$ the $k$-dimensional vector space whose entries are elements of $H^2$. If $\Tilde{P}$ is the orthogonal projection of $(L^2)^k$ onto $(H^2)^k$ and $f\in (L^\infty)^{k\times k}$, then $T(f)$ is the matrix representation of the operator
\begin{equation*}
    (H^2)^k\rightarrow(H^2)^k,\quad \vect g \mapsto \Tilde{P}(f\vect g),
\end{equation*}
where $f\vect g$ is the usual matrix-vector multiplication. Note that the $(H^2)^k$ equivalent of Theorem \ref{thm: F. and M. Riesz} holds verbatim with $H^2$ replaced by $(H^2)^k$. With this in mind, we are now able to prove a generalization of Coburn's lemma (for the Toeplitz case, see \cite[Theorem 1.10]{Douglas1998}).

\begin{theorem}[Coburn's lemma; tridiagonal $k$-Toeplitz version]\label{thm: coburns lemma tridiagonal}
   Let $f \in \C^{k\times k}(\mathbb{T})$ be the symbol of a tridiagonal $k$-Toeplitz operator such that $\det(f(z))$ does not vanish identically on $\mathbb{T}$. Then one of the following statements holds:
   \begin{enumerate}[(i)]
       \item $T(f)$ has a trivial kernel;
       \item $T(f)$ has a dense range;
       \item The leading $(k-1)\times(k-1)$ principal minor of $A_0$ has a nonzero kernel. In particular, there exists some $z_0 \in \C\cup\{\infty\}$ such that $\ker(z_0^{-1}A_{-1}+A_0)\cap\ker(A_1)\neq \vect 0$.
   \end{enumerate}
\end{theorem}

\begin{proof}
    It is not hard to see that the adjoint of $T(f)$ is $T(\Bar{f}^\top)$, where the superscript $\top$ denotes the transpose. Assume that $T(f)$ has a nontrivial kernel and a nondense range, i.e., conditions (i) and (ii) do not hold. This implies that $T(\Bar{f}^{\top})$ has a nontrivial kernel as well. Hence there exist nonzero functions
    $\vect g_+, \vect h_+ \in (H^2)^k$ such that
    \begin{equation}\label{equ:proofcoburnlem1}
        f\vect g_+ = \vect g_- \in (H^2_-)^k \quad \text{and}\quad \Bar{f}^\top \vect h_+ = \vect h_- \in (H^2_-)^k,
    \end{equation} 
    where $\vect g_{\pm} = (g_{\pm,1}, g_{\pm,2}, \cdots, g_{\pm,k})^\top$ and $\vect h_{\pm} = (h_{\pm,1}, h_{\pm,2},\cdots, h_{\pm, k})^\top$. 
    Since $\vect g_+, \vect h_+ \in (H^2)^k$, there exist two sequences of vectors $(\vect v_n)_n, (\vect w_n)_n \in \C^k$ such that
    \begin{equation}\label{equ:proofcoburnlem0}
    \begin{aligned}
        \vect g_+ &= \sum_{n\in \N} \vect v_n  z^{n-1},\\
        \vect h_+ &= \sum_{n\in \N} \vect w_n z^{n-1},
    \end{aligned} 
    \end{equation}
    for $z=e^{i\theta}$, where we have absorbed the factor $\frac{1}{\sqrt{2\pi}}$ into $\vect v_n$ and $\vect w_n$. 
    This allows us to write
    \begin{equation}\label{equ:proofcoburnlem2}
    \begin{aligned}
        \vect g_- &= z^{-1}A_{-1}\vect v_1 + (A_0\vect v_1 + A_{-1}\vect v_2)\\
        &\phantom{=}
        + \sum_{n\in\N}z^n(A_{-1}\vect v_{n+2} + A_0 \vect v_{n+1} + A_1\vect v_n)\\
        &=z^{-1}A_{-1}\vect v_1,
    \end{aligned}
    \end{equation}
    where we have used that $\vect g_-\in (H^2_-)^k$ in the last equality. Similarly, using that $\vect h_- \in (H^2_-)^k$, we obtain the following expression for $\vect h_-$:
    \begin{equation}\label{equ:proofcoburnlem3}
    \begin{aligned}
        \vect h_- =& z^{-1}\Bar{A_1}^\top \vect w_1 + (\Bar{A_0}^\top \vect w_1 + \Bar{A_1}^\top \vect w_2)\\
        &+ \sum_{n\in\N}z^n(\Bar{A_1}^\top \vect w_{n+2} + \Bar{A_0}^\top \vect w_{n+1} + \Bar{A_{-1}}^\top\vect w_n)\\
        =&z^{-1}\Bar{A_1}^\top \vect w_1.
    \end{aligned}
    \end{equation}
    By Theorem \ref{thm: F. and M. Riesz}, it follows that $\vect g_+, \vect h_+ \neq 0$ almost everywhere on $\mathbb T$. We define
    \begin{equation*}
        \varphi :=  \Bar{\vect h}_+^\top \vect g_- =  \Bar{\vect h}_+^\top f\vect g_{+}= \Bar{(\Bar{f}^{\top}{\vect h}_+)}^{\top} \vect g_{+}=\Bar{\vect h}_-^\top \vect g_+.
    \end{equation*}
    It holds that $\varphi \in L^1$. Moreover, $\varphi_n = (\Bar{\vect h}_-^\top \vect  g_+)_n = 0$ for $n\leq 0$ and $\varphi_n = (\Bar{\vect h}_+^\top \vect g_-)_n = 0$ for $n\geq 0$ since by (\ref{equ:proofcoburnlem0}), (\ref{equ:proofcoburnlem2}), and (\ref{equ:proofcoburnlem3}),
    \begin{align*}
        \Bar{\vect h}_+^\top \vect g_- &= \sum_{n \in \N} z^{-n} \Bar{\vect w}_n^\top A_{-1}\vect v_1,\\
        \Bar{\vect h}_-^\top \vect  g_+ &= \sum_{n \in \N} z^{n} (\Bar{\Bar{A_{1}}^\top \vect w}_1)^\top \vect v_n.
    \end{align*}
    Hence $\varphi = 0$. This implies $\Bar{A_{1}}^\top \vect w_1 \perp \vect v_n$ and $A_{-1} \vect v_1 \perp \vect w_n$ for $n \in \N$. Thus either $\vect v_n \perp \vect e_k, n=1,\cdots$, or $\vect w_1 \in \ker \bar{A_{1}}^\top$. However, if $\vect w_1 \in \ker \bar{A_{1}}^\top$, then by (\ref{equ:proofcoburnlem3})
    \begin{align*}
        \Bar{f}^\top\vect h_+ = \vect h_-
        =z^{-1}\Bar{A_1}^\top \vect w_1
        =\vect 0,
    \end{align*}  
    which contradicts the fact that $\det(\Bar{f}^\top) = \Bar{\det(f)}$ does not vanish identically on $\mathbb{T}$. Therefore, $\vect v_n \perp \vect e_k$ for $n\in\N_{\geq1}$. By the above arguments, we conclude that if $T(f)$ has non-dense range and nontrivial kernel, then $\vect v_n \perp \vect e_k$ for $n\in\N_{\geq1}$ and for any $\vect v_n$ defined by (\ref{equ:proofcoburnlem0}). 
    
    Now we prove that if $T(f)$ has a non-dense range and a nontrivial kernel, then the condition (iii) must hold. Since $f\vect g_+ = \vect g_- \in (H^2_-)^k$, by expansion (\ref{equ:proofcoburnlem2}) it holds that $A_{-1}\vect v_{n+2} + A_0 \vect v_{n+1} + A_1\vect v_n =\vect 0$ for $n\in\N$. Since $\vect v_n \perp \vect e_k$ for $n\in\N_{\geq1}$ it follows that $\vect v_n \in \ker(A_1)$. Hence it holds that $A_{-1}\vect v_{n+1} + A_0\vect v_{n} = \vect 0$ for $n\in\N_{\geq2}$. Note that the expansion (\ref{equ:proofcoburnlem2}) also gives $A_{-1}\vect v_2+A_0\vect v_1=0$. Therefore, $A_{-1}\vect v_{n+1} + A_0\vect v_{n} = \vect 0$ for $n\in\N_{\geq1}$. Furthermore, since the image of $A_{-1}$ is confined to the span of $\vect e_k$ and $\vect v_n \perp \vect e_k$ for $n\in\N_{\geq1}$, we have 
    \[
   A_{0}|_{(k-1)\times (k-1)} \vect v_{n}|_{1:k-1} =\vect 0,
    \]
    where $A_{0}|_{(k-1)\times (k-1)}$ is
    the leading principal $(k-1)\times (k-1)$ submatrix of $A_0$. This also means that $A_{0}|_{(k-1)\times (k-1)}$ cannot have full rank. On the other hand, one can compute that
    \begin{equation}\label{eq:proofcoburnlem4}
        (z^{-1}A_{-1}+A_0)\vect v_n =
        \begin{pmatrix}
            0\\
            \vdots\\
            0\\
            c_{k-1}\vect v_n^{(k-1)} + z^{-1}b_k\vect v_n^{(1)}
        \end{pmatrix},\quad n\geq 1
    \end{equation}
    where $\vect v_n^{(j)}$ is the $j$-th element of $\vect v_n$. Considering the case when $n=1$, by (\ref{equ:proofcoburnlem2}), we know that $\vect v_1^{(1)}\neq 0$, otherwise, 
    \[
    f \vect g_+ = \vect g_{-} = z^{-1}A_{-1}\vect v_1 =\vect 0,
     \]
    which contradicts $\det f$ does not vanish identically on $\mathbb{T}$. Similarly, we can also show that $b_k\neq 0$. Thus by (\ref{eq:proofcoburnlem4}) and $b_k\neq 0$ there exists some $z_0 \in \C\cup\{\infty\}$ such that $\vect v_1$ lies in the kernel of $z_0^{-1}A_{-1}+A_0$. Therefore it holds that 
    $\ker(z_0^{-1}A_{-1} + A_0)\cap\ker(A_1) \neq \vect 0$, which concludes the proof. 
\end{proof}

\subsection{Spectra of tridiagonal $k$-Toeplitz operators}
In this section, we will provide a full characterization of the spectrum of tridiagonal $k$-Toeplitz operators. We begin by first recalling (and generalizing) a few theorems on Toeplitz operators. The first important theorem is a consequence of \cite[Chapter 23, Theorem 4.3]{gohberg1993basic}. Its counterpart for Toeplitz operators was already established in \cite[Chapter 23, Corollary 4.4]{gohberg1993basic}.

We denote by $\sigma(\cdot), \sigma_{ess}(\cdot)$  the spectrum and the essential spectrum of the operator, respectively. We define 
\begin{equation}\label{equ:detspectraformula1}
\sigma_{\mathrm{det}}(f) = \{\lambda\in \mathbb C: \det(f(z)-\lambda)=0,\ \exists  z\in \mathbb T \}.
\end{equation}
We first have the following theorem characterizing the essential spectrum of $T(f)$.

\begin{theorem}\label{thm: essential spectrum}
The operator $T(f)$ in (\ref{equ:tridiagonalktoeplitzop2}) is Fredholm if and only if
\(\mathrm{det}(f)\) has no zeros on \(\mathbb{T}\). Furthermore, the essential spectrum of $T(f)$ is given by
    \begin{equation*}
     \sigma_{ess}(T(f)) = \sigma_{\mathrm{det}}(f).
    \end{equation*}
\end{theorem}
\begin{proof}
The first part has been proven in \cite[Chapter 23, Theorem 4.3]{gohberg1993basic}. For the second part, note that $\lambda \in \sigma_\text{ess}(T(f)) $ if and only if $T(f)- \lambda$ is not a Fredholm operator. By the first part, this happens only if $\det(f(z_0)-\lambda)=0$ for some $z_0 \in \mathbb{T}$. That is, $\lambda$ is an eigenvalue of $f(z_0)$. On the other hand, if $\lambda$ is an eigenvalue for some $z_0 \in \mathbb{T}$, it holds that $\det(f(z_0) -\lambda) = 0$, which concludes the proof.
\end{proof}

We now recall the following Theorem due to Gohberg \cite[Chapter 23, Theorem 
5.1]{gohberg1993basic}. 
\begin{theorem}\label{thm: windingnumber}
    Let $f:\mathbb{T}\to \C^{k\times k}$ be such that \(T(f)\) is a Toeplitz operator. Then,  \(T(f)\) is Fredholm on the space \(\ell^2\) if and only if \(\mathrm{det}(f)\) has no zeros on \(\mathbb{T}\), in which case
    \begin{equation*}
        \operatorname{Ind}T(f) = -\operatorname{wind}(\det(f(\mathbb T)),0),
    \end{equation*}
    with $\operatorname{wind}(\det(f(\mathbb T)),0)$ being the winding number about the origin of the determinant of $f(z)$ with $z\in \mathbb T$.  
\end{theorem}
Utilizing this theorem together with Theorem \ref{thm: coburns lemma tridiagonal}, we are now able to establish one of the main results of our paper, an optimal characterization of the spectra of tridiagonal $k$-Toeplitz operators. Before stating the theorem, we first define 
\begin{equation}\label{equ:windspectraformula1}
\sigma_{\mathrm{wind}}(f):= \{\lambda \in \mathbb{C}\setminus \sigma_{\mathrm{det}}(f) :  \operatorname{wind}(\det(f(\mathbb T) - \lambda),0) \neq 0 \}.
\end{equation}
By Lemma \ref{lem:tridiagonaldeterminant}, 
\[
\det(f(z) - \lambda) = (-1)^{k+1}\left(\prod_{i = 1}^k c_i\right) z + (-1)^{k+1}\left(\prod_{i = 1}^k b_i\right) z^{-1} + g(\lambda)
\]
with $g(\lambda)$ being given by (\ref{equ:defiglambda1}), which yields
\begin{equation}\label{equ:windspectraformula2}
\sigma_{\mathrm{wind}}(f)= \left\{\lambda \in \mathbb{C}\setminus \sigma_{\mathrm{det}}(f):  \operatorname{wind}\left((-1)^{k+1} \left((\prod_{i = 1}^k c_i) z + (\prod_{i = 1}^k b_i) z^{-1}\right), -g(\lambda)\right) \neq 0 \right\}.
\end{equation}
Furthermore, since 
\[
\mathrm{det}(f(z)-\lambda) = \prod_{j = 1}^k (\lambda_j(z)-\lambda), 
\]
where $\lambda_j(z), 1\leq j\leq k,$ are the eigenvalues of the matrix $f(z)$, we obtain
\begin{align}
\sigma_{\mathrm{wind}}(f)=& \left\{\lambda \in \mathbb{C}\setminus \sigma_{\mathrm{det}}(f) :  \sum_{j=1}^k\operatorname{wind}(\lambda_j(\mathbb T) - \lambda,0) \neq 0 \right\} \nonumber \nonumber  \\
=& \left\{\lambda \in \mathbb{C}\setminus \sigma_{\mathrm{det}}(f) :  \sum_{j=1}^k\operatorname{wind}(\lambda_j(\mathbb T), \lambda) \neq 0 \right\}. \label{equ:windspectraformula3}
\end{align}

Note that the representations (\ref{equ:windspectraformula1}), (\ref{equ:windspectraformula2}) and (\ref{equ:windspectraformula3}) have their own pro and contra. For relating our results to the conjecture in \cite{dimerSkin}, we utilize representation (\ref{equ:windspectraformula3}) in our main results. We now establish the following theorem for the spectrum of the operator $T(f)$. 
\begin{theorem}\label{thm: windingspectrum}
    Let $f \in \C^{k\times k}(\mathbb{T})$ be the symbol of a tridiagonal $k$-Toeplitz operator $T(f)$. Denote by $B_0$ the leading $(k-1)\times(k-1)$ principal minor of $A_0$. It holds that
  \[
  \sigma_{\mathrm{det}}(f)\cup \sigma_{\mathrm{wind}}(f)  \subset \sigma(T(f)) \subset \sigma_{\mathrm{det}}(f)\cup \sigma_{\mathrm{wind}}(f) \cup \sigma(B_0),
  \]
  where $\sigma_{\mathrm{det}}(f), \sigma_{\mathrm{wind}}(f)$ are given by (\ref{equ:detspectraformula1}), (\ref{equ:windspectraformula3}), respectively. 
\end{theorem}
\begin{proof}
    The first inclusion is an immediate consequence of Theorems \ref{thm: essential spectrum} and \ref{thm: windingnumber}. For the second inclusion, note that if
    \begin{equation}\label{equ:proofwindingspectrum1}
        \lambda \in \sigma(T(f))\setminus \left(\sigma_{\mathrm{det}}(f)\cup \sigma_{\mathrm{wind}}(f) \right),
    \end{equation}
     then it must hold that $T(f)-\lambda I$ is Fredholm and $\operatorname{Ind}(T(f)-\lambda I) = \operatorname{wind}(\det(f-\lambda),0)=0$. Based on the definition (\ref{defi:fredholmindex1}), this implies
     \begin{equation}\label{equ:proofwindingspectrum2}
     \mathrm{dim} \ \mathrm{Ker} (T(f)-\lambda I) = \mathrm{dim} (\ell^2\setminus \mathrm{Im} (T(f)-\lambda I)).
     \end{equation}
     Moreover, since $\lambda \in \sigma(T(f))$, $T(f)-\lambda I$ has a non-trivial kernel and non-dense image. By Theorem \ref{thm: coburns lemma tridiagonal}, this implies that $\lambda$ must be an eigenvalue of $B_0$.
\end{proof}

This characterization is optimal. It cannot be made more precise, as it cannot be guaranteed that an eigenvalue $\lambda$ of $B_0$ with $\operatorname{wind}(\det(f-\lambda),0)=0$ is also an eigenvalue of $T(f)$. To illustrate this, consider
\begin{equation*}
    f(z) = 
    \begin{pmatrix}
        0 & 1 + \frac{1}{2}z \\
        1 + \frac{1}{2}z^{-1}&1\\
    \end{pmatrix}.
\end{equation*}
The determinant of this symbol is given by $\det(f)=-(1+\frac{1}{2}z)(1+\frac{1}{2}z^{-1})$. We have $\operatorname{wind}(\det(f(\mathbb T)),0)=0$ and $\det f(z)\neq 0, \forall z\in \mathbb T$. In this case, $B_0 = 0$ is just the top left entry of $f(z)$. Furthermore, by (\ref{defi:fredholmindex1}), if $0\in \sigma(T(f))$, then the kernel space of $T(f)$ must be non-trivial. Since 
\begin{equation*}
    T(f) = 
    \begin{pmatrix}
        0 & 1 & 0 & 0& \\
        1 & 1 & \frac{1}{2}& 0 & \\
        0 & \frac{1}{2} & 0 & 1 & \\
        0 & 0 & 1 & 1 & \\
        &&&&\ddots
    \end{pmatrix},
\end{equation*}
one can check that, after scaling, the only possible eigenvector of $T(f)$ corresponding to the eigenvalue $0$ is $\vect u = (1,0,-2,0,4,\dots)^\top$. However, $\vect u$ does not lie in $\ell^2$, which means that $0\not \in \sigma(T(f))$. 


On the other hand, consider
\begin{equation*}
    \Tilde{f}(z) = 
    \begin{pmatrix}
        0 & 1 + 2z \\
        1 + 2z^{-1}&1\\
    \end{pmatrix}.
\end{equation*}
The determinant is $\det(\Tilde{f})=-(1+2z)(1+2z^{-1})$, hence $0 \notin \sigma_{\mathrm{det}}(f)$ and $\operatorname{wind}(\det(\Tilde{f}(\mathbb T)),0)=0$. Therefore, by Theorem \ref{thm: coburns lemma tridiagonal}, the operator $T(\Tilde{f})$ is either invertible or has a nontrivial kernel. By explicitly writing down $T(\Tilde{f})$
\begin{equation*}
    T(\Tilde{f}) = 
    \begin{pmatrix}
        0 & 1 & 0 & 0& \\
        1 & 1 & 2& 0 & \\
        0 & 2 & 0 & 1 & \\
        0 & 0 & 1 & 1 & \\
        &&&&\ddots
    \end{pmatrix},
\end{equation*}
we find that $\Tilde{\vect u} = (1,0,-\frac{1}{2},0,\frac{1}{4},\dots)^\top$ satisfies $T(\Tilde{f})\Tilde{\vect u} = \vect 0$. Then in this case, $\Tilde{\vect u} \in \ell^2$ and $0\in \sigma(T(f))$.

\begin{remark}
Note that (\ref{equ:windspectraformula2}) and (\ref{equ:windspectraformula3}) also provide ways to compute $\sigma_{\mathrm{wind}}(f)$ explicitly. We will utilize (\ref{equ:windspectraformula3}) in Section \ref{Sec: Topological_origin} for numerical illustrations of the topological origin of the skin effect. 
\end{remark}

For the sake of completeness, we will devote the last part in this section to tridiagonal $k$-Laurent operators that are defined as
\begin{equation}\label{equ:tridiagonalklaurentop2}
    L(f) = \begin{pmatrix}
    \ddots & \ddots & \ddots &  \\
       & A_1& A_0 & A_{-1} &  \\
       & & A_{1} & A_0 & A_{-1}&  \\
        & & & \ddots & \ddots &\ddots  \\
    \end{pmatrix},
\end{equation}
with the generating symbol $f$ being (\ref{eq: symbol tridiagonal operator}). Using the identification $L(f-\lambda) = L(f)-\lambda$, which we have used to characterize the spectrum of tridiagonal $k$-Toeplitz operators, one can obtain an explicit 
characterization of the spectrum of tridiagonal $k$-Laurent operators as a consequence of \cite[Chapter 23, Corollary 2.5]{gohberg1993basic}. The following theorem holds. 

\begin{theorem}
    Let $f \in (\mathcal{L}(\mathbb{T})^\infty)^{k\times k}$ and denote the associated tridiagonal $k$-Laurent operator $L(f)$. Then, it holds that
    \begin{equation*}
        \sigma(L(f)) =\sigma_{\text{ess}}(L(f)) =  \sigma_{\det}(f)
    \end{equation*}
    with $\sigma_{\det}(f)$ being defined by (\ref{equ:detspectraformula1}).
\end{theorem}

\subsection{Eigenvectors of tridiagonal $k$-Toeplitz operators}

Now that we have characterized the spectra of tridiagonal $k$-Toeplitz operators, we will venture onwards to explore their eigenvectors. The following theorem is the second main result of this paper.

\begin{theorem}\label{thm: exponential_decay_k_operators}
        Suppose $\Pi_{j=1}^k c_j\neq 0$ and $\Pi_{j=1}^k b_j\neq 0$. Let $f(z) \in \mathbb{C}^{k\times k}$ be the symbol (\ref{eq: symbol tridiagonal operator}) and let $\lambda \in \C\setminus\sigma_{ess}(T(f))$. If 
        $\sum_{j=1}^{k}\operatorname{wind}(\lambda_j(\mathbb T), \lambda)<0$ for $\lambda_j$ being the eigenvalues of  $f(z)$, then there exists an eigenvector $\bm x\in\ell^2$ of $T(f)$ associated to $\lambda$ and some $\rho<1$ such that 
    \begin{equation}\label{eq: bound on eigenvector}
        \frac{\lvert \bm x_j\rvert}{\max_{i}\lvert \bm x_i\rvert} \leq C \lceil j/k\rceil \rho^{\lceil j/k\rceil-1},\quad j\geq 1,
    \end{equation}
   where $C>0$ is a constant depending only on $\lambda, a_p, b_p, c_p$ for $1\leq p\leq k$. 
\end{theorem}
\begin{proof}
    Let \(\lambda \in \mathbb{C}\) be such that \(\sum_{j=1}^{k}\operatorname{wind}(\lambda_j(\mathbb T), \lambda) = \operatorname{wind}(\det(f(\mathbb T)-\lambda),0)<0\). Since $\prod_{i=1}^kb_i \neq 0, \prod_{i=1}^kc_i \neq 0$, by Lemma \ref{lem:tridiagonaldeterminant}, \(\det(f(z)-\lambda)\) is a meromorphic function in $\mathbb C$, with poles at $0$ and $\infty$. 
    Then the argument principle implies that \(\det(f(z)-\lambda) = 0\) for two values \(z_1,z_2\) with $1<\frac{1}{\rho}\leq \babs{z_j}, j=1,2$, counted with multiplicity. Moreover, since \(\det(f(z_i)-\lambda) = 0\), we find that \(\lambda \) is an eigenvalue to \(f(z_i)\) and hence there exists an associated eigenvector \(\mathbf{v}_i\). Assume for now that \(z_1 \neq z_2\). Consider the vectors defined by
    \begin{equation}
        \mathbf{u}_i= (\mathbf{v}_i^{\top}, z_i^{-1} \mathbf{v}_i^{\top}, z_i^{-2} \mathbf{v}_i^{\top}, \dots)^\top,  \quad i \in \{1,2\}.
    \end{equation}
    The vector $\mathbf{u}_i$ satisfies the eigenvalue equation \(T(f)\mathbf{u}_i = \lambda\mathbf{u}_i\) in all but the first row. Since $\vect u_1, \vect u_2$ are linearly independent, there exists a linear combination of them that is indeed an eigenvector of \(T(f)\).

    Next, we consider the case where $z_1 = z_2$ and first suppose that $\dim\operatorname{Eig}(f(z_1),\lambda)>1$, where $\operatorname{Eig}(f(z_1),\lambda)$ denotes the eigenspace of $f(z_1)$ associated to the eigenvalue $\lambda$. Then there exist linearly independent eigenvectors $\vect v_1, \vect v_2$ of $f(z_1)$ associated to $\lambda$ and we can use
    \begin{equation}
        \mathbf{u}_i= (\mathbf{v}_i^{\top}, z_1^{-1} \mathbf{v}_i^{\top}, z_1^{-2} \mathbf{v}_i^{\top}, \dots)^\top \quad i \in \{1,2\},
    \end{equation}
    to construct an exponentially decaying eigenvector of $T(f)$. Finally, it remains to treat the case when $z_1 = z_2$ and $\dim\operatorname{Eig}(f(z_1),\lambda)=1$. Denote $\vect v_1$ the eigenvector of $f(z_1)$ associated to $\lambda$. We will construct a vector $\vect v_2$ such that 
    \begin{equation}
        \vect u_2 = (\vect v_2^{\top},\ z_1^{-1} \vect v_2^{\top} + \vect v_1^{\top},\ z_1^{-2} \vect v_2^{\top} + 2z_1^{-1} \vect v_1^{\top},\ \dots)^\top
    \end{equation}
   satisfies the eigenvalue problem in all but the first row. Similarly to the previous arguments, this is enough to construct the eigenvector of $T(f)$. By explicitly writing down $(T(f)-\lambda)\vect u_2=0$ (except for the first row), we find that $\vect v_2$ must satisfy the condition
   \begin{equation}\label{eq: condition}
       (A_{1} + (A_0 - \lambda)z_1^{-1} + A_{-1} z_1^{-2})\vect v_2 = -((A_0 - \lambda) + 2A_{-1}z_1^{-1})\vect v_1,
   \end{equation}
   which is also a sufficient condition. From Lemma \ref{lem:tridiagonaldeterminant}, it is not hard to see that for all $\Tilde{\lambda}$ in a small neighborhood of $\lambda$, the equation $\det(f(z)-\Tilde{\lambda})$ has two distinct roots. Let $(\epsilon_j)_j$ be a sequence converging to zero. We define $\lambda(\epsilon_j) = \lambda + \epsilon_j$. Then by the previous observation, we may assume that to each $\lambda(\epsilon_j)$ there exist two distinct roots $z_1(\epsilon_j), z_2(\epsilon_j)$ of the equation $\det(f(z)-\lambda(\epsilon_j))=0$. Let $\vect w_1(\epsilon_j), \vect w_2(\epsilon_j)$ be the corresponding unit eigenvectors. We may assume that
   \begin{equation}
       \vect w_1(\epsilon_j)\rightarrow \vect v_0,\quad  \vect w_2(\epsilon_j) \rightarrow \vect v_0.
   \end{equation}
   The limit of $\vect w_i(\epsilon_j)$ (or a subsequence thereof) exists by a standard compactness argument and it must be an eigenvector of $f(z_1)$ associated to $\lambda$. In particular, for each $j$ it holds that
   \begin{align*}
    &\left(A_{1}+(A_0-\lambda(\epsilon_j))z_1(\epsilon_j)^{-1}+A_{-1} z_1(\epsilon_j)^{-2}\right) \vect {w_1}(\epsilon_j) =0,\\
    &\left(A_{1}+(A_0-\lambda(\epsilon_j))z_2(\epsilon_j)^{-1}+A_{-1} z_2(\epsilon_j)^{-2}\right) \vect {w_2}(\epsilon_j) =0.
    \end{align*}
    We define
    \begin{equation*}
        \Delta z(\epsilon_j) = z_1(\epsilon_j)^{-1} - z_2(\epsilon_j)^{-1}.
    \end{equation*}
    Then we have
    \begin{align*} 
        &A_{1} \frac{\vect w_1(\epsilon_j)-\vect w_2(\epsilon_j)}{\Delta z(\epsilon_j)}+ (A_0-\lambda(\epsilon_j))\frac{z_1(\epsilon_j)^{-1}\vect w_1(\epsilon_j)-z_2(\epsilon_j)^{-1}\vect w_2(\epsilon_j)}{\Delta z(\epsilon_j)}\\
        & \qquad +A_{-1}\frac{z_1(\epsilon_j)^{-2}\vect w_1(\epsilon_j)-z_2(\epsilon_j)^{-2}\vect w_2(\epsilon_j)}{\Delta z(\epsilon_j)} =0.
    \end{align*}    
    Further expansion of this expression then yields
    \begin{equation}\label{eq: proofdecay1}
    \begin{aligned}
        &A_{1}\frac{\vect w_1(\epsilon_j)-\vect w_2(\epsilon_j)}{\Delta z(\epsilon_j)}+ (A_0-\lambda(\epsilon_j))\frac{z_1(\epsilon_j)^{-1}(\vect w_1(\epsilon_j)-\vect w_2(\epsilon_j))}{\Delta z(\epsilon_j)}+A_{-1} \frac{z_1(\epsilon_j)^{-2}(\vect w_1(\epsilon_j)-\vect w_2(\epsilon_j))}{\Delta z(\epsilon_j)}\\
        =&- \left((A_0-\lambda(\epsilon_j))\frac{(z_1(\epsilon_j)^{-1}-z_2(\epsilon_j)^{-1})\vect w_2(\epsilon_j))}{\Delta z(\epsilon_j)}+A_{-1} \frac{(z_1(\epsilon_j)^{-2}-z_2(\epsilon_j)^{-2})\vect w_2(\epsilon_j)}{\Delta z(\epsilon_j)}\right).
    \end{aligned} 
    \end{equation}
    Consider now the sequence $\frac{\vect w_1(\epsilon_j) - \vect w_2(\epsilon_j)}{\Delta z(\epsilon_j)}$. If a bounded subsequence exists, we may define $\vect v_2$ as the limit of this subsequence. In this case, it is not hard to see that $\vect v_2$ satisfies the condition (\ref{eq: condition}). Otherwise, there exists a subsequence of  $\frac{\vect w_1(\epsilon_j) - \vect w_2(\epsilon_j)}{\Delta z(\epsilon_j)}$ whose norm tends to infinity. Define 
    \begin{equation*}
        \Delta \vect w(\epsilon_j) = \lVert \vect w_1(\epsilon_j) - \vect w_2(\epsilon_j) \rVert.
    \end{equation*}
    Then (upon passing to said subsequence) it holds that
    \begin{equation}
        \frac{\Delta z(\epsilon_j)}{\Delta \vect w(\epsilon_j)} \rightarrow 0.
    \end{equation}
    Multiplying (\ref{eq: proofdecay1}) with this yields
    \begin{align*}
        &A_{1} \frac{\vect w_1(\epsilon_j)-\vect w_2(\epsilon_j)}{\Delta \vect w(\epsilon_j)}+ (A_0-\lambda(\epsilon_j))\frac{z_1(\epsilon_j)^{-1}(\vect w_1(\epsilon_j)-\vect w_2(\epsilon_j))}{\Delta\vect w(\epsilon_j)}+A_{-1} \frac{z_1(\epsilon_j)^{-2}(\vect w_1(\epsilon_j)-\vect w_2(\epsilon_j))}{\Delta\vect w(\epsilon_j)}\\
        =&- \left((A_0-\lambda(\epsilon_j))\frac{(z_1(\epsilon_j)^{-1}-z_2(\epsilon_j)^{-1})\vect w_2(\epsilon_j))}{\Delta\vect w(\epsilon_j)}+A_{-1} \frac{(z_1(\epsilon_j)^{-2}-z_2(\epsilon_j)^{-2})\vect w_2(\epsilon_j)}{\Delta\vect w(\epsilon_j)}\right).
    \end{align*}
    In the limit, the right-hand side tends to $0$. Thus $\vect v_2 = \lim_{\epsilon_j \to 0} \frac{\vect w_1(\epsilon_j) - \vect w_2(\epsilon_j)}{\Delta \vect w(\epsilon_j)}$ satisfies
    \begin{equation*}
        \left( A_{1} + (A_0-\lambda)z_1^{-1} + A_{-1} z_1^{-2}\right)\vect v_2 = 0.
    \end{equation*}
    This indicates that $\vect v_2$ is an eigenvector. But from the construction of $\vect v_2$, it becomes apparent that $\vect v_2$ must be orthogonal to $\vect v_1$, which contradicts the assumption that the eigenspace of $f(z_1)$ associated with $\lambda$ is one-dimensional. Therefore, the sequence $\frac{\vect w_1(\epsilon_j) - \vect w_2(\epsilon_j)}{\Delta z(\epsilon_j)}$ is uniformly bounded and arguments for the first case already demonstrate the existence of such $\vect v_2$ satisfying (\ref{eq: condition}).  

    The exponential bound in (\ref{eq: bound on eigenvector}) is clear by the construction of the vectors $\vect u_1$ and $\vect u_2$.
\end{proof}

\begin{remark}
Note that the constructive approach used in the proof of \cite[Theorem 7.2]{trefethen.embree2005Spectra} does not work here. We have used a new approximation approach to construct the eigenvectors. 
\end{remark}

\begin{remark}
For the case when $\sum_{j=1}^{k}\operatorname{wind}(\lambda_j(\mathbb T), \lambda)>0$, $\lambda \in \sigma(T(f))$ is because $T(f)$ has a non-dense image. In particular, one can construct the eigenvector $\bm x$ of $\overline{T(f)}^{\top}-\overline{\lambda}$ in the same spirit of Theorem \ref{thm: exponential_decay_k_operators}, which implies $T(f)-\lambda$ has a non-dense image.  
\end{remark}



%% file: Pseudo_spectra.tex
\section{Pseudo-spectra of tridiagonal $k$-Toeplitz matrices}\label{sec: spectral_Theory_matrix}
In this section, we will apply the results of the previous section to matrices. As already described by Trefethen in \cite{REICHEL1992153}, the spectrum of non-Hermitian matrices is very sensitive to small perturbations. For this reason, it is often advisable to study their pseudo-spectra, rather than their spectra. For the convenience of the reader, we recall the definition of pseudoeigenvalues and pseudoeigenvectors.
\begin{definition}\label{def: pseudospectrum}
    Let \(\varepsilon > 0\). Then, \(\lambda \in \mathbb{C}\) is an \(\varepsilon\)-pseudoeigenvalue of \(\mathbf{A}\in \mathbb{C}^{N\times N}\) if one of the following conditions is satisfied:
    \begin{enumerate}[(i)]
        \item \(\lambda\) is a proper eigenvalue of \(\mathbf{A} + \mathbf{E}\) for some \(\mathbf{E} \in \mathbb{C}^{N\times N}\) such that \(\lVert \mathbf{E} \rVert \leq \varepsilon\);
        \item $\lVert (\mathbf{A}-\lambda I) \mathbf{u} \rVert < \varepsilon$ for some vector $u$ with $\lVert \vect u \rVert = 1$;
        \item $\lVert (\mathbf{A}- \lambda I)^{-1} \rVert^{-1} \leq \epsilon$.
    \end{enumerate}
    The set of all $\varepsilon$-pseudoeigenvalues of $\vect A$, the $\varepsilon$-pseudospectrum, is denoted by $\sigma_{\epsilon}(A)$. If some nonzero $\vect u$ satisfies $\lVert (\lambda - \mathbf{A}) \mathbf{u} \rVert < \varepsilon$, then we say that $\vect u$ is an $\varepsilon$-pseudoeigenvector of $\mathbf{A}$.
\end{definition}
By the equivalence of norms in finite dimensions, any norm can be used. For a general treatment of pseudospectra, see \cite{trefethen.embree2005Spectra}.

We consider the family of tridiagonal $k$-Toeplitz matrices $\left\{\mathbf{A}_N\right\}$ of various dimensions obtained as $N \times N$ finite sections of the infinite operator $T(f)$. Results on the pseudospectra of the family of Toeplitz matrices can be found in \cite{REICHEL1992153}. On the other hand, the spectrum of tridiagonal $k$-Toeplitz matrices is fully understood and closed form characteristic polynomials have been found in \cite{da2007characteristic}.

Let us also introduce the sets
\begin{align}
    \Omega_r &:= \{ \lambda \in \C: \operatorname{wind}(\operatorname{det}(f(\mathbb{T}_r)- \lambda), 0) > 0 \} \label{eq: omega_r},\\
    \Omega^R &:= \{ \lambda \in \C: \operatorname{wind}(\operatorname{det}(f(\mathbb{T}_R)- \lambda), 0) < 0 \}\label{eq: omega^r}, 
\end{align}
where $\mathbb{T}_r = \{ z \in \C: \lvert z \rvert = r \}$ and let $\sigma_\epsilon(\mathbf{A}_N)$ be the set of $\epsilon$-pseudoeigenvalues of $\mathbf{A}_N$.

The connection between the eigenvectors of a tridiagonal $k$-Toeplitz operator and those of the associated tridiagonal $k$-Toeplitz matrix can be made with an argument similar to \cite[Theorem 7.2]{trefethen.embree2005Spectra}. The first $N$ components of an eigenvector of the Toeplitz operator can be used as the $\epsilon$-pseudoeigenvector of the tridiagonal $k$-Toeplitz matrix $\mathbf{A}_N$.
Since we have proven the existence of an exponentially decaying eigenvector in Theorem \ref{thm: exponential_decay_k_operators}, we can conclude with the following statement on tridiagonal $k$-Toeplitz matrices.
 \begin{theorem}\label{thm: exponential_decay_tridiag_k_matrices}
    Let \(\{\mathbf{A}_N\}\) be a family of tridiagonal $k$-Toeplitz matrices such that $\prod_{i=1}^kb_i \neq 0$, $
    \prod_{i=1}^kc_i \neq 0$. Then, for any $r<1$ and $\rho>r$, we have
\begin{equation}\label{eq:pseudoeigenvalue1}
\Omega_r \cup \Omega^{1 / r} \cup\left(\sigma(T(f))+\Delta_{\varepsilon}\right) \subseteq \sigma_{\varepsilon}\left(A_N\right) \quad \text { for } \quad \varepsilon=\max(C_1, \lceil N/k\rceil C_2)\rho^{\lceil N/k \rceil-1},
\end{equation}
where $\Delta_{\epsilon}$ is a closed unit disk of radius $\epsilon$ and $C_1, C_2$ are constants depending on $r, \rho$, and $a_j, b_j, c_j$ for $1\leq j\leq k$ but independent of $N$. In particular, for the corresponding $\lambda\in \Omega_r \cup \Omega^{1/r}$, there exist nonzero pseudoeigenvectors \(\mathbf{v}^{(N)}\) satisfying
     \begin{equation}\label{eq: bound eigenvector2}
         \frac{\lVert (\mathbf{A}_N-\lambda)\mathbf{v}^{(N)}\rVert}{\lVert \mathbf{v}^{(N)}\rVert} \leq \max(C_3, \lceil N/k\rceil C_4)\rho^{\lceil N/k \rceil-1}
     \end{equation}
     such that
    \begin{equation}\label{equ:exponentialdecayvector2}
         \frac{|\vect v_j^{(N)}|}{\max_{i}|\vect v_i^{(N)}|} \leq
        \begin{cases}
            C_5\lceil N/k\rceil \rho^{\lceil j/k \rceil-1}, & \text{if } 
            \lambda \in \Omega^{1/r}, \\
             C_5\lceil N/k\rceil\rho^{\lceil (N-j)/k\rceil-1}, & \text{if } \lambda \in \Omega_{r},
         \end{cases}
         \quad 1\leq j \leq N,
     \end{equation}
     where $C_3, C_4, C_5$ are constants independent of $N$.
 \end{theorem}
 
 \begin{proof}
Firstly, we note that the inclusion $\sigma(T(f))+\Delta_{\varepsilon} \subseteq \sigma_{\varepsilon}\left(A_N\right)$ is triviality valid for any matrix or operator. Secondly, by symmetry, $\Omega_{r}$ must satisfy an estimate of the type (\ref{eq:pseudoeigenvalue1}) if $\Omega^{1/r}$ does. Thus all that we have to prove is $\Omega^{1/r} \subseteq \sigma_{\varepsilon}\left(\vect A_N\right)$.

    Given any $r<1$, let $\lambda \in \Omega^{1/r}$ be arbitrary. Since $\operatorname{wind}(\det(f(\mathbb T_{\frac{1}{r}})-\lambda),0) < 0$, similarly to the arguments in the proof of Theorem \ref{thm: exponential_decay_k_operators}, we can show that \(\det(f(z)-\lambda) = 0\) for two values \(z_1,z_2\) with $1<\frac{1}{\rho}\leq \babs{z_j}, j=1,2$, counted with multiplicity. Also, since \(\det(f(z_i)-\lambda) = 0\), we find that \(\lambda \) is an eigenvalue to \(f(z_i)\) and hence there exists an associated eigenvector \(\mathbf{v}_i\). Assume for now that \(z_1 \neq z_2\). Consider the vectors defined by
    \begin{equation}\label{equ:proofpesuospectra1}
        \mathbf{u}_i= (\mathbf{v}_i^{\top}, z_i^{-1} \mathbf{v}_i^{\top}, z_i^{-2} \mathbf{v}_i^{\top}, \dots, z_{i}^{-\lceil N/k \rceil+1}\vect v_i^{\top})^\top \quad i \in \{1,2\},
    \end{equation}
    where without loss of generality we assume that $N$ is divisible by $k$.
    This vector satisfies the eigenvalue equation \(\vect A_N \mathbf{u}_i = \lambda\mathbf{u}_i\) in all but the first and last rows. Since $\vect u_1, \vect u_2$ are linearly independent, there exists a normalized linear combination $\vect v^{(N)}= a_1\vect u_1+ a_2 \vect u_2$ so that $(\vect A_N-\lambda) \vect v^{(N)} = \vect 0$ in all but the last row. It is not hard to see that 
    \begin{equation}
        \lVert (\vect A_N - \lambda)\vect v^{(N)}\rVert \leq C_3\max \{ \babs{z_1}^{-\lceil N/k \rceil+1}, \babs{z_2}^{-\lceil N/k \rceil+1} \},
    \end{equation}
    where $C_3$ depends on $a_j, b_j, c_j, \lambda$ but not on $N$. Then since $1<\frac{1}{\rho}\leq \babs{z_j}, j=1,2,$ we obtain the bound
    \begin{equation}
        \lVert (\vect A_N - \lambda)\vect v^{(N)}\rVert \leq C_3 \rho^{\lceil N/k \rceil-1},
    \end{equation}
    which proves the claim (\ref{eq: bound eigenvector2}) if $z_1 \neq z_2$. If the roots $z_1, z_2$ are distinct for all $\lambda \in \overline{\Omega^{1/r}}$, then the function 
    \[
\sup _{N \geq 1} \frac{\sigma_N\left(\lambda I-\vect A_N\right)}{r^{\lceil N/k \rceil-1}}
\]
is a continuous function of $\lambda$ on the compact set $\overline{\Omega^{1/r}}$; its maximum provides the constant $C_1$ in (\ref{eq:pseudoeigenvalue1}) that independent of $\lambda, N$, and we can take $\rho = r$.  
    
In the case where $z_1=z_2$, one can resort to a similar argument as the one in Theorem \ref{thm: exponential_decay_k_operators} to utilize
\begin{equation}\label{equ:proofpesuospectra2}
\begin{aligned}
\mathbf{u}_1&= (\mathbf{v}_1^{\top},\ z_1^{-1} \mathbf{v}_1^{\top},\ z_1^{-2} \mathbf{v}_1^{\top},\ \dots,\ z_{1}^{-\lceil \frac{N}{k} \rceil+1}\vect v_1^{\top})^\top,\\
\vect u_2 &= (\vect v_2^{\top},\ z_1^{-1} \vect v_2^{\top} + \vect v_1^{\top},\ z_1^{-2} \vect v_2^{\top} + 2z_1^{-1} \vect v_1^{\top},\ \dots,\ z_{1}^{-\lceil \frac{N}{k} \rceil+1}\vect v_2^{\top}+ (\lceil \frac{N}{k} \rceil-1)z_{1}^{-\lceil \frac{N}{k} \rceil+2}\vect v_1^{\top})^\top
\end{aligned}
\end{equation}
to construct the pseudoeigenvector. This yields a similar bound except with $C_3\rho^{\lceil N/k \rceil-1}$ in (\ref{eq: bound eigenvector2}) replaced by an algebraically growing factor at worst $\lceil N/k\rceil C_4\rho^{\lceil N/k \rceil-1}$. Then since the function 
    \[
\sup _{N \geq 1} \frac{\sigma_N\left(\lambda I-\vect A_N\right)}{\lceil N/k\rceil r^{\lceil N/k \rceil-1}}
\]
is a continuous function of $\lambda$ on the compact set $\overline{\Omega^{1/r}}$, its maximum provides the constant $C_2$ in (\ref{eq:pseudoeigenvalue1}). Finally, (\ref{equ:exponentialdecayvector2}) is deduced from construction (\ref{equ:proofpesuospectra1}) and (\ref{equ:proofpesuospectra2}). 
\end{proof}

%% file: Topological_origin.tex

\section{Topological origin of the skin effect in polymer systems of subwavelength resonators}\label{Sec: Topological_origin}

The non-Hermitian skin effect is the phenomenon whereby in the subwavelength regime a large proportion of  the bulk eigenmodes of a non-Hermitian open system of subwavelength resonators are localised at one edge of the structure \cite{ashida.gong.ea2020NonHermitian, okuma.kawabata.ea2020Topological}. It has been realized experimentally in topological photonics, phononics, and other condensed matter systems \cite{ghatak.brandenbourger.ea2020Observation, okuma.kawabata.ea2020Topological, longhi}. Its significance  lies in its substantial contribution to advancing the realm of active metamaterials, paving the way for novel opportunities to guide and control energy on subwavelength scales \cite{ammari.barandun.ea2023Edge, ammari.davies.ea2021Functional, ammari.davies.ea2020Topologically, feppon.cheng.ea2023Subwavelength}.

Recently, the non-Hermitian skin effect has been demonstrated in one-dimensional systems of subwavelength resonators using first-principle mathematical analysis. In particular, \cite{ammari2023mathematical} demonstrates the skin effect for the monomer case 
when an imaginary gauge potential is introduced inside the resonators to break Hermiticity and \cite{dimerSkin} generalizes the results in \cite{ammari2023mathematical} to dimer systems. We also refer to \cite{3DSkin} for the skin effect in a three-dimensional setting and to \cite{stabilitySkin} for stability analysis. 

By the results obtained in the previous sections, we provide the topological origin of the skin effect in the one-dimensional polymer system, i.e., in systems with periodically repeated cells of $k$ resonators. This also elucidates the numerical findings in \cite{dimerSkin} for the topological origin of the skin effect.

Let us first introduce the model setting. We consider a chain of $N$ disjoint one-dimensional resonators \(D := (x_i^{\mathrm{L}}, x_i^{\mathrm{R}}\)), where \((x_i^{\mathrm{L}, \mathrm R})_{1 \leq i \leq N} \subset \R\) are the $2N$ extremities such that $x_i^{\mathrm{L}} < x_i^{\mathrm{R}} < x_{i+1}^{\mathrm{L}}$ for any $1 \leq i \leq N$. The lengths of the $i$-th resonators will be denoted by $\ell_i = x_i^{\mathrm{R}}-x_i^{\mathrm{L}}$ and the spacings between the $i$-th and the $(i+1)$-th resonators will be denoted by $s_i = x_{i+1}^{\mathrm{L}}-x_i^{\mathrm{R}}$. Furthermore, we assume that $k$ resonators are periodically repeated, that is, $s_{i+k} = s_i$. The resulting system is illustrated in Figure \ref{fig:setting}.

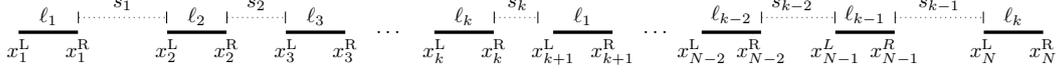
\begin{figure}[htb]
    \centering
    \begin{adjustbox}{width=\textwidth}
    \begin{tikzpicture}
        \coordinate (x1l) at (1,0);                 
        \path (x1l) +(1,0) coordinate (x1r);        
        \path (x1r) +(0.75,0.7) coordinate (s1);
        \path (x1r) +(1.5,0) coordinate (x2l);      
        \path (x2l) +(1,0) coordinate (x2r);
        \path (x2r) +(0.5,0.7) coordinate (s2);
        \path (x2r) +(1,0) coordinate (x3l);        
        \path (x3l) +(1,0) coordinate (x3r);
        \path (x3r) +(1,0.7) coordinate (s3);
        \path (x3r) +(1.5,0) coordinate (x4l);
        \path (x3r) +(0.75,0) coordinate (dots1);   
        \path (x4l) +(1,0) coordinate (x4r);        
        \path (x4r) +(0.4,0.7) coordinate (s4);
        \path (x4r) +(1,0) coordinate (x5l);        
        \path (x5l) + (1,0) coordinate (x5r);
        \path (x5r) + (0.75,0) coordinate (dots2);  
        \path (dots2) +(0.75,0) coordinate (x6l);   
        \path (x6l) +(1,0) coordinate (x6r);
        \path (x6r) +(1.25,0) coordinate (x7l);     
        \path (x6r) +(0.525,0.7) coordinate (s6);
        \path (x7l) +(1,0) coordinate (x7r);
        \path (x7r) +(1.5,0) coordinate (x8l);      
        \path (x7r) +(0.75,0.7) coordinate (s7);
        \path (x8l) +(1,0) coordinate (x8r);
        \draw[ultra thick] (x1l) -- (x1r);      
        \node[anchor=north] (label1) at (x1l) {$x_1^{\mathrm{L}}$};
        \node[anchor=north] (label1) at (x1r) {$x_1^{\mathrm{R}}$};
        \node[anchor=south] (label1) at ($(x1l)!0.5!(x1r)$) {$\ell_1$};
        \draw[dotted,|-|] ($(x1r)+(0,0.25)$) -- ($(x2l)+(0,0.25)$);
        \draw[ultra thick] (x2l) -- (x2r);      
        \node[anchor=north] (label1) at (x2l) {$x_2^{\mathrm{L}}$};
        \node[anchor=north] (label1) at (x2r) {$x_2^{\mathrm{R}}$};
        \node[anchor=south] (label1) at ($(x2l)!0.5!(x2r)$) {$\ell_2$};
        \draw[dotted,|-|] ($(x2r)+(0,0.25)$) -- ($(x3l)+(0,0.25)$);
        \draw[ultra thick] (x3l) -- (x3r);      
        \node[anchor=north] (label1) at (x3l) {$x_3^{\mathrm{L}}$};
        \node[anchor=north] (label1) at (x3r) {$x_3^{\mathrm{R}}$};
        \node[anchor=south] (label1) at ($(x3l)!0.5!(x3r)$) {$\ell_3$};
        \node (dots1) at (dots1) {\dots};
        \draw[ultra thick] (x4l) -- (x4r);      
        \node[anchor=north] (label1) at (x4l) {$x_k^{\mathrm{L}}$};
        \node[anchor=north] (label1) at (x4r) {$x_k^{\mathrm{R}}$};
        \node[anchor=south] (label1) at ($(x4l)!0.5!(x4r)$) {$\ell_k$};
        \draw[dotted,|-|] ($(x4r)+(0,0.25)$) -- ($(x5l)+(-.25,0.25)$);
        \draw[ultra thick] (x5l) -- (x5r);       
        \node[anchor=north] (label1) at (x5l) {$x_{k+1}^{\mathrm{L}}$};
        \node[anchor=north] (label1) at (x5r) {$x_{k+1}^{\mathrm{R}}$};
        \node[anchor=south] (label1) at ($(x5l)!0.5!(x5r)$) {$\ell_{1}$};
        \node (dots2) at (dots2) {\dots};
        \draw[ultra thick] (x6l) -- (x6r);      
        \node[anchor=north] (label1) at (x6l) {$x_{N-2}^{\mathrm{L}}$};
        \node[anchor=north] (label1) at (x6r) {$x_{N-2}^{\mathrm{R}}$};
        \node[anchor=south] (label1) at ($(x6l)!0.5!(x6r)$) {$\ell_{k-2}$};
        \draw[dotted,|-|] ($(x6r)+(0,0.25)$) -- ($(x7l)+(0,0.25)$);
        \draw[ultra thick] (x7l) -- (x7r);      
        \node[anchor=north] (label1) at (x7l) {$x_{N-1}^{L}$};
        \node[anchor=north] (label1) at (x7r) {$x_{N-1}^{R}$};
        \node[anchor=south] (label1) at ($(x7l)!0.5!(x7r)$) {$\ell_{k-1}$};
        \draw[dotted,|-|] ($(x7r)+(0,0.25)$) -- ($(x8l)+(0,0.25)$);
        \draw[ultra thick] (x8l) -- (x8r);      
        \node[anchor=north] (label1) at (x8l) {$x_{N}^{\mathrm{L}}$};
        \node[anchor=north] (label1) at (x8r) {$x_{N}^{\mathrm{R}}$};
        \node[anchor=south] (label1) at ($(x8l)!0.5!(x8r)$) {$\ell_{k}$};
        \node[anchor=north] (label1) at (s1) {$s_1$};
        \node[anchor=north] (label1) at (s2) {$s_2$};
        \node[anchor=north] (label1) at (s4) {$s_k$};
        \node[anchor=north] (label1) at (s6) {$s_{k-2}$};
        \node[anchor=north] (label1) at (s7) {$s_{k-1}$};
    \end{tikzpicture}
    \end{adjustbox}
    \caption{A chain of $N$ times periodically repeated $k$ subwavelength resonators, with lengths
    $(\ell_i)_{1\leq i\leq N}$ and spacings $(s_{i})_{1\leq i\leq N-1}$.}
    \label{fig:setting}
\end{figure}
\noindent
Let $D$ be the domain consisting of the union of all resonators
\begin{equation}
    D = \bigcup_{i =1}^N (x_i^{\mathrm{L}}, x_i^{\mathrm{R}}) \subset \R.
\end{equation}
The underlying wave equation that governs the system is given by 
\begin{equation}\label{eq: wave equation}
\begin{cases}u^{\prime \prime}(x)+\gamma u^{\prime}(x)+\frac{\omega^2}{v^2} u=0, & x \in D, \\ u^{\prime \prime}(x)+\frac{\omega^2}{v_b^2} u=0, & x \in \mathbb{R} \backslash D, \\ \left.u\right|_{\mathrm{R}}\left(x_i^{\mathrm{L}, \mathrm{R}}\right)-\left.u\right|_{\mathrm{L}}\left(x_i^{\mathrm{L}, \mathrm{R}}\right)=0, & \text { for all } 1 \leq i \leq N, \\ \left.\frac{\mathrm{d} u}{\mathrm{~d} x}\right|_{\mathrm{R}}\left(x_i^{\mathrm{L}}\right)=\left.\delta \frac{\mathrm{d} u}{\mathrm{~d} x}\right|_{\mathrm{L}}\left(x_i^{\mathrm{L}}\right), & \text { for all } 1 \leq i \leq N, \\  \left.\frac{\mathrm{d} u}{\mathrm{~d} x}\right|_{\mathrm{L}}\left(x_i^{\mathrm{R}}\right)=\left. \delta \frac{\mathrm{d} u}{\mathrm{~d} x}\right|_{\mathrm{R}}\left(x_i^{\mathrm{R}}\right), & \text { for all } 1 \leq i \leq N, \\ \frac{\mathrm{d} u}{\mathrm{~d}|x|}-\mathrm{i} \frac{\omega}{v} u=0, & \text{ at } x= x_1^{\mathrm L}, x_N^{\mathrm R}, \text{ respectively},
\end{cases}
\end{equation}
where $0 < \delta \ll 1$, the wave speeds inside the resonator and in the background are respectively $v_b$ and $v$ and are of order one, and $\gamma u'$ for $\gamma \in \mathbb{R}^*$ denotes the imaginary gauge potential. Here, for a function $w$ we denote by
$$ \left. w \right|_{\mathrm{L}} (x) := \lim_{s\rightarrow 0, s>0}
w(x-s) \text{ and }  \left. w \right|_{\mathrm{R}} (x) := \lim_{
s\rightarrow 0, s>0} w(x +s)$$
if the limits exist.

We are interested in nontrivial frequencies $\omega$ such that (\ref{eq: wave equation}) is satisfied and $\omega \rightarrow 0$ as $\delta\rightarrow 0$. Such $\omega$ are called subwavelength resonances.

In the high-contrast, low-frequency setting described above, gauge capacitance matrices provide a discrete approximation of the subwavelength eigenfrequencies induced in the structure.
\begin{definition}
    For $\gamma \in \mathbb{R}^*$, we define the \emph{gauge capacitance matrix} $C^\gamma \in \mathbb{R}^{N\times N}$ by
\begin{equation}\label{capactitance matrix definition}
    C^\gamma_{i,j} = \begin{cases} \frac{\gamma}{s_1}\frac{\ell_1}{1-e^{-\gamma \ell_1}}, & i = j = 1, \\
    \frac{\gamma}{s_i}\frac{\ell_i}{1-e^{-\gamma \ell_i}} - \frac{\gamma}{s_{i-1}} \frac{\ell_i}{1-e^{\gamma \ell_i}}, & 1 < i = j < N, \\
    -\frac{\gamma}{s_i}\frac{\ell_i}{1-e^{-\gamma \ell_j}}, & 1 \leq i = j-1  \leq N-1, \\
    \frac{\gamma}{s_j}\frac{\ell_i}{1-e^{\gamma \ell_j}}, & 2 < i = j+1 \leq N, \\
    -\frac{\gamma}{s_{N-1}} \frac{\ell_N}{1-e^{\gamma \ell_N}}, & i = j = N,
    \end{cases},
\end{equation}
and all the other entries are zero. 
\end{definition}
In \cite[Corollary 2.6]{ammari2023mathematical}, one can find a formal argumentation that the gauge capacitance matrix gives indeed a valid approximation of the eigenmodes in a subwavelength regime as $\delta \rightarrow 0$.
For the remainder of this paper, we assume that $\gamma = \ell_i = 1$. Without loss of generality, we consider a system of $nk$ resonators and obtain a gauge capacitance matrix of the tridiagonal $k$-Toeplitz form (perturbed):
\begin{equation}\label{eq: k-capacitance matrix nk+1}
    \mathbf{A}^{(a,b)}_{nk} = \begin{pmatrix}
        a_1 + a & b_1 & 0 &\cdots  & \cdots & 0 \\
        c_1 & a_2 & b_2 & \ddots  &  & \vdots  \\
        0 & \ddots &  \ddots & \ddots & \ddots & \vdots  \\
        \vdots & \ddots&c_{k-3} & a_{k-2} & b_{k-2} & 0 \\
        \vdots& &  \ddots & c_{k-2} & a_{k-1}& b_{k-1}  \\
        0 & \cdots & \cdots & 0 & c_{k-1} & a_k+ b 
    \end{pmatrix}.
\end{equation}
For the semi-infinite systems we obtain a perturbed tridiagonal $k$-Toeplitz operator 
\[
T^{a}(f) = \begin{pmatrix}
        a_1 + a & b_1 &   \\
        c_1 & a_2 & b_2 &   \\
         & \ddots &  \ddots & \ddots  \\
        & &c_{k-1} & a_{k} & b_{k}  \\
        & &   & c_{k} & a_{1}& b_{1}  \\
         & & &  & \ddots & \ddots& \ddots
        \end{pmatrix}
\]
with $f$ being the symbol (\ref{eq: symbol tridiagonal operator}).

It has been observed and conjectured \cite{dimerSkin} that the exponential decay of the eigenvectors of $k$-Toeplitz matrices for $k \geq 2$ is due to the winding of the eigenvalues of the symbol being nontrivial. Now, by the spectral theory for the tridiagonal $k$-Toeplitz operators introduced in the previous sections, we have the following results for the topological origin of the skin effect in  subwavelength resonator systems, validating the conjecture made in \cite{dimerSkin}. 
\begin{theorem}\label{thm: exponential_decay_capacitance_operators}
        Suppose $\Pi_{j=1}^k c_j\neq 0$ and $\Pi_{j=1}^k b_j\neq 0$. Let $f(z) \in \mathbb{C}^{k\times k}$ be the symbol (\ref{eq: symbol tridiagonal operator}) and let $\lambda \in \C\setminus\sigma_{ess}(T^{a}(f))$. If $\sum_{j=1}^k\operatorname{wind}(\lambda_j(\mathbb T), \lambda)<0$, then there exists an eigenvector $\bm x$ of $T^{a}(f)$ associated to $\lambda$ and some $\rho<1$ such that 
    \begin{equation}\label{eq: bound on eigenvectorN}
        \frac{\lvert \bm x_j\rvert}{\max_{i}\lvert \bm x_i\rvert} \leq C \lceil j/k\rceil \rho^{\lceil j/k\rceil-1},
    \end{equation}
   where $C>0$ is a constant depending only on $\lambda, a_j, b_j, c_j, j=1,\cdots, k$. If $\sum_{j=1}^k\operatorname{wind}(\lambda_j(\mathbb T), \lambda)>0$, then the above results hold for the left eigenvectors. 
\end{theorem}
\begin{proof} Although with a perturbation $a$ on the first element, the proof of Theorem \ref{thm: exponential_decay_k_operators} can still be applied. 
\end{proof}

Theorem \ref{thm: exponential_decay_capacitance_operators} elucidates the topological origin of the skin effect in the polymer system of subwavelength resonators. In particular, the skin effect holds for all $\lambda$ in the region 
\begin{equation}\label{equ:nonzerowindingregion1}
G:=\left\{ \lambda \in \mathbb{C}\setminus \sigma_{\mathrm{det}}(f) :\sum_{j=1}^k\operatorname{wind}(\lambda_j(\mathbb T), \lambda)\neq 0 \right\}.
\end{equation}
This is a generalization for the topological origin of the skin effect in the Toeplitz operator case given in \cite{ammari2023mathematical}.



The last part of the section is devoted to illustrating numerically the skin effect and its topological origin in chains of $2$ and $3$ periodically repeated resonators.
We start by illustrating in Figure \ref{Fig: spectrum_dimer_system} the results of a 
system of $2$ periodically repeated resonators as in \cite{dimerSkin}. In Figure \ref{Fig: spectrum_dimer_system}A,
we show the spectrum and pseudospectrum of the gauge capacitance matrix of a system of $25$ dimers together with the winding of the two eigenvalues of the symbol of the corresponding 2-Toeplitz operator. Figure \ref{Fig: spectrum_dimer_system}B shows that all the eigenvectors (black eigenvectors) associated with eigenvalues inside the region $G$ in (\ref{equ:nonzerowindingregion1}) are localised and the only non-decreasing eigenvector (gray eigenvector) corresponds to the eigenvalue $0$ in the boundary of the region $G$. On the other hand, in \cite{dimerSkin} it is observed and conjectured that the non-trivial winding of the eigenvalues $\lambda_j(z)$ predicts the exponential decay of the eigenmodes. This is due to  $\operatorname{wind}(\lambda_j(\mathbb T), \lambda)\leq 0, j=1,2$ in the example, which yields 
\[
G = \left\{ \lambda \in \mathbb{C}\setminus \sigma_{\mathrm{det}}(f) :  \bigcup_{j=1}^2 \operatorname{wind}(\lambda_j(\mathbb T), \lambda)\neq 0\right\}.  
\]

\begin{figure}[h]
    \centering
    \subfloat[][The region of non-trivial winding of the \\eigenvalues and the pseudo-spectrum.]{{\includegraphics[width=6.3cm]{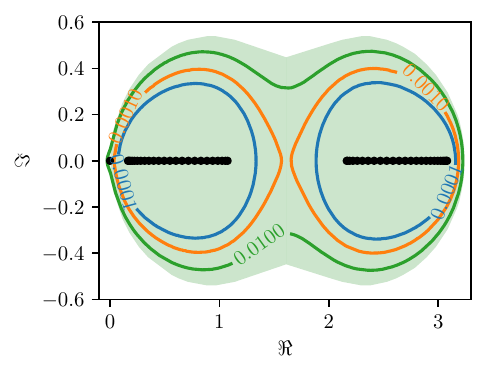} }}
    \qquad
    \subfloat[][Eigenmodes superimposed on one another to portray the skin effect. ]{{\includegraphics[width=6.3cm]{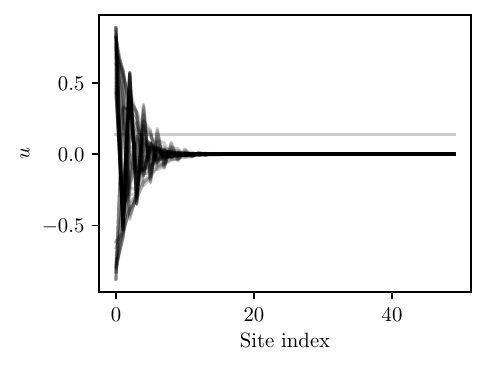} }}%
    \caption{The region of $\lambda$ so that $\sum_{j=1}^k\operatorname{wind}(\lambda_j(\mathbb T), \lambda)\neq 0$ and the localization of the eigenvectors. Computation performed for $s_1 = 1, s_2 = 2,$ and $N = 50$.}
   \label{Fig: spectrum_dimer_system}
\end{figure}

%

\begin{figure}[h]
    \centering
    \subfloat[][Computation performed for $s_1 = 1, s_2 = 2, s_3 = 3,$ and $N = 50$. ]{{\includegraphics[width=6.3cm]{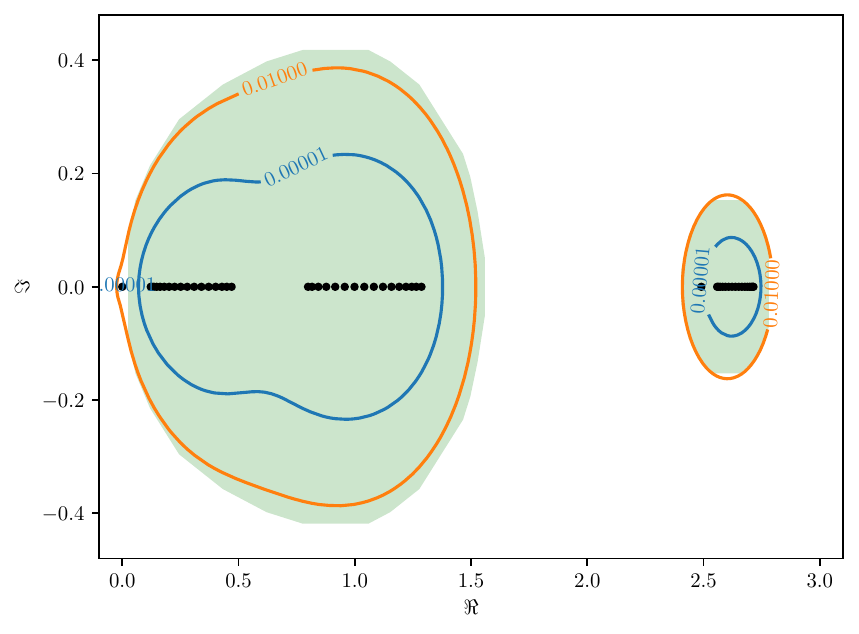} }}
    \qquad
    \subfloat[][ Simulation performed with $s_1 = 1, s_2 =2, s_3 = 3,$ and $N = 50$.]{{\includegraphics[width=6.3cm]{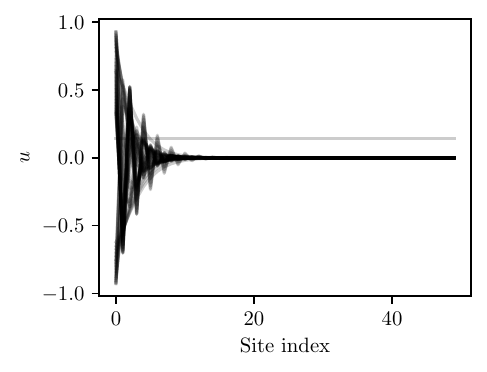} }}
    \qquad
    \subfloat[][Computation performed for $s_1 = 2, s_2 = 3, s_3 = 4,$ and $N = 50$. ]{{\includegraphics[width=6.3cm]{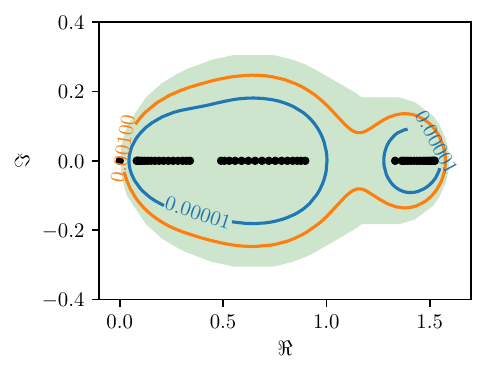} }}%
    \qquad
    \subfloat[][Simulation performed with $s_1 = 2, s_2 =3, s_3 = 4,$ and $N = 50$.]{{\includegraphics[width=6.3cm]{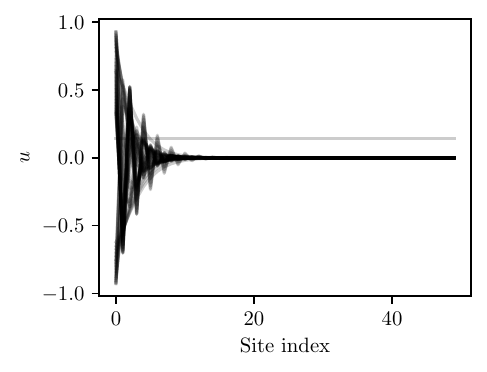} }}%
    \caption{Figures A and C show the spectrum of the operator. The green regions consist of all the eigenvalues $\lambda$ that satisfy $\sum_{j=1}^3 \operatorname{wind}(\lambda_j(\mathbb T), \lambda)\neq 0$. The black dots along the real line denote the spectrum of the gauge capacitance matrix $C^\gamma$ and the solid blue and orange lines around the spectrum are the $\varepsilon$-pseudospectra for $\varepsilon = 10^{k}$ and $k = -5, -2$. Figures B and D show the eigenvectors of $C^{\gamma}$.}
   \label{Fig: spectrum_trimer_system}
\end{figure}

Figure \ref{Fig: spectrum_trimer_system} illustrates the results for $3$ periodically repeated resonators. We numerically verify that indeed all the eigenvectors, except the one associated with eigenvalue $0$ on the boundary of $G$, are localised at the left edge of the structure.




\begin{remark}
We remark that the green regions for $k=1,2,3,4$ can be computed analytically through the formula (\ref{equ:windspectraformula3}). For the case when $k\geq 5$, one can utilize (\ref{equ:windspectraformula3}) to compute numerically the green region. 
\end{remark}

\section{Concluding remarks}\label{sec: concluding_remarks}
We have developed new theories of spectra and pseudo-spectra for tridiagonal $k$-Toeplitz operators and matrices. Specifically, we have established the relationship between the winding number of the symbol function's eigenvalues and the exponential decay property of the corresponding eigenvectors (or pseudo-eigenvectors). This discovery sheds light on the topological origin of the non-Hermitian skin effect observed in one-dimensional polymer systems of subwavelength resonators. This paper also opens the door to the study of the spectral theory of block Toeplitz operators. 

%% file: Appendix.tex
\appendix

\section{Computation of the determinant}\label{appendix: expansion_of_determinant}
\begin{lemma}\label{lem:tridiagonaldeterminant}
    For $\lambda \in \C\setminus\sigma_{ess}(T(f))$, the determinant $\det(f(z)-\lambda)$, where $f$ is the symbol of the tridiagonal Toeplitz operator $T$ given in (\ref{eq: symbol tridiagonal operator}), has the form
    \begin{equation}\label{equ:detexpansion1}
        \det(f(z)-\lambda) = (-1)^{k+1}(\prod_{i = 1}^k c_i) z + (-1)^{k+1}(\prod_{i = 1}^k b_i) z^{-1} + g(\lambda),
    \end{equation}
    where $g$ is a polynomial in $\lambda$ of degree $k$, defined by (\ref{equ:defiglambda1}). In particular, there are at most $2k$ $\lambda\in \mathbb C$ so that $\det(f(z)-\lambda) =0$ admits a double root. 
\end{lemma} 
\begin{proof}
The case when $k=1,2$ is easy to verify. For the case when $k\geq3$, by repeated Laplace expansion, one can show that
    \begin{align*}
        \det(f(z)-\lambda) &= (-1)^{k+1}b_1 b_k 
        \begin{vmatrix}
            b_2 & 0 & \dots & \dots & 0 \\
            a_3-\lambda & \ddots & \ddots & & \vdots \\
            c_3 & \ddots & \ddots & \ddots & \vdots\\
            \vdots & \ddots & \ddots & \ddots & 0 \\
            0 & \dots & c_{k-2} & a_{k-1}-\lambda & b_{k-1}\\
        \end{vmatrix}
        z^{-1} \\
        &\qquad\qquad+ (-1)^{k+1} c_kc_{k-1}
        \begin{vmatrix}
            c_1 & a_2 - \lambda & b_2 &0 & \dots & 0 \\
            0 & c_2 & a_3 - \lambda & b_3 & \ddots & \vdots \\
            \vdots & \ddots & \ddots & \ddots & \ddots & 0 \\
            \vdots  & & \ddots&\ddots & \ddots & b_{k-3} \\
            \vdots & &  & \ddots & c_{k-3} & a_{k-2} - \lambda  \\
            0 & \dots&\dots & \dots & 0& c_{k-2}  \\
        \end{vmatrix}
        z
        \end{align*}
        \begin{align*}
        &+ (a_k-\lambda)
        \begin{vmatrix}
            a_1 - \lambda & b_1 & 0 & \dots & 0 \\
            c_1 & a_2 - \lambda & b_2 & \ddots & \vdots \\
            0 & \ddots & \ddots & \ddots & 0 \\
            \vdots & \ddots& \ddots & \ddots & b_{k-1} \\
            0 & \dots &  0 & c_{k-2} & a_{k-1} - \lambda \\
        \end{vmatrix}
        \\
        &\qquad\qquad-b_{k-1}c_{k-1}
        \begin{vmatrix}
            a_1 - \lambda & b_1 & 0 & \dots & 0 \\
            c_1 & a_2 - \lambda & b_2 & \ddots & \vdots \\
            0 & \ddots & \ddots & \ddots & 0 \\
            \vdots & \ddots & \ddots & \ddots & b_{k-3} \\
            0 & \dots & 0 & c_{k-3} & a_{k-2} - \lambda \\
        \end{vmatrix}
        \\
        &\qquad\qquad\qquad\qquad-b_kc_k
        \begin{vmatrix}
            a_2 - \lambda & b_2 & 0 & \dots & 0 \\
            c_3 & a_3 - \lambda & b_3 & \ddots & \vdots \\
            0 & \ddots & \ddots & \ddots & 0 \\
            \vdots & \ddots & \ddots & \ddots & b_{k-2} \\
            0 & \dots & 0 & c_{k-1} & a_{k-1} - \lambda \\
        \end{vmatrix}.
    \end{align*}
    Hence (\ref{equ:detexpansion1}) holds for
    \begin{equation}\label{equ:defiglambda1}
        g(\lambda) = \det(A_0-\lambda)  -b_kc_k p(\lambda)  ,
    \end{equation}
    where
    \[
p(\lambda) = \begin{cases}
0, & k=1, \\
1, & k=2,\\
\begin{vmatrix}
            a_2 - \lambda & b_2 & 0 & \dots & 0 \\
            c_3 & a_3 - \lambda & b_3 & \ddots & \vdots \\
            0 & \ddots & \ddots & \ddots & 0 \\
            \vdots & \ddots & \ddots & \ddots & b_{k-2} \\
            0 & \dots & 0 & c_{k-1} & a_{k-1} - \lambda \\
        \end{vmatrix}, & k \geq 3.
\end{cases}
    \]
     To demonstrate the last claim, we write
     \begin{equation}\label{equ:proofeigenvector1}
         \det(f(z)-\lambda) = (Az + Bz^{-1}) + g(\lambda)
     \end{equation}
     with $A = (-1)^{k+1}(\prod_{i = 1}^k c_i), B = (-1)^{k+1}(\prod_{i = 1}^k b_i)$. By multiplication with $z$, we obtain that
     \[
      z\det(f(z)-\lambda) = Az^2 + B + g(\lambda)z.
     \]
     Thus only when
     \[
     g(\lambda)^2-4AB =0, 
     \]
     we have a double root. Since $g(\lambda)$ is a polynomial of order $k$, we have at most $2k$ solutions $\lambda$ of the above equation so that the root is a double root.
\end{proof}